\documentclass[11pt, reqno]{amsart}

\usepackage{multicol} 

\usepackage{amsmath,amsthm,amssymb,mathrsfs,stmaryrd,color}

\usepackage[top=1in, bottom=1in, left=0.5in, right=0.5in]{geometry}
\usepackage{float}
\usepackage{graphicx}

\usepackage[utf8]{inputenc}
\usepackage[T1]{fontenc}

\usepackage{enumitem}
\setlist[enumerate]{itemsep=2pt,parsep=2pt,before={\parskip=2pt}}

\usepackage[colorlinks=true,hyperindex, linkcolor=red!60, pagebackref=false, citecolor=cyan, pdfpagelabels]{hyperref}
\usepackage[capitalize]{cleveref}

\usepackage{url}
\usepackage{breakurl}

\usepackage{algorithm,algorithmicx}
\usepackage{algpseudocode}

\newcommand*\Let[2]{\State #1 $\gets$ #2}
\algnewcommand\algorithmicinput{\textbf{Input:}}
\algnewcommand\algorithmicinputphantom{\phantom{\textbf{Input:}}}
\algnewcommand\Input{\item[\algorithmicinput]}
\algnewcommand\PhantomInput{\item[\algorithmicinputphantom]}

\usepackage{tikz}
\usetikzlibrary{calc}
\usetikzlibrary{math}
\usetikzlibrary{snakes}

\usepackage{thmtools} 
\usepackage{thm-restate}

\newtheorem{theorem}{Theorem}[section]
\newtheorem*{theorem*}{Theorem}
\newtheorem*{definition*}{Definition}
\newtheorem{proposition}[theorem]{Proposition}
\newtheorem{lemma}[theorem]{Lemma}

\theoremstyle{definition}
\newtheorem{definition}[theorem]{Definition}

\newtheorem{remark}[theorem]{Remark}

\newtheorem{claim}[theorem]{Claim}

\newcommand{\Z}{\ensuremath{\mathbb{Z}}}
\newcommand{\N}{\ensuremath{\mathbb{N}}}

\newcommand{\dil}{\mathrm{dil}}
\newcommand{\dilt}{\mathrm{dil}_3}

\newcommand{\opti}{\mathcal{M}}
\newcommand{\optiloc}{\mathcal{M}_{\mathrm{loc}}}

\begin{document}

\title{A note on optimal degree-three spanners of the square lattice}

\author{Damien Galant}
\address{Department of Mathematics, University of Mons (UMONS), Place du Parc 20, 7000 Mons, Belgium.}
\email{damien.galant@student.umons.ac.be}

\author{Cédric Pilatte}
\address{Department of Mathematics and their Applications, \'Ecole Normale Supérieure (ENS), Rue d'Ulm, 45, 75005 Paris, France --- Department of Mathematics, University of Mons (UMONS), Place du Parc 20, 7000 Mons, Belgium.}
\email{cedric.pilatte@ens.fr}

\begin{abstract}
In this short note, we prove that the degree-three dilation of the square lattice $\mathbb{Z}^2$ is $1+\sqrt{2}$. 
This disproves a conjecture of Dumitrescu and Ghosh. We give a computer-assisted proof of a local-global 
property for the uncountable set of geometric graphs achieving the optimal dilation.
\end{abstract}

\vspace*{-0.35cm} 

\maketitle

\begin{multicols}{2}

\section{Introduction}

Let $P$ be a set of points in the Euclidean plane. A \emph{geometric graph} $G$ on $P$ is an undirected graph drawn in the plane whose vertices are the points of $P$ and whose edges are straight line segments between the corresponding points. We write $d_G(p, q)$ for the length of the shortest path between $p$ and $q$ that uses only edges of $G$ ($+\infty$ if there is no such path). A geometric graph is \emph{plane} if no two edges intersect (except possibly at a common vertex). 

We measure the efficiency of a geometric graph $G$ with a real number, called the \emph{dilation} (or \emph{stretch factor}, or \emph{spanning ratio}) of $G$. The \emph{dilation of a pair} $(p, q)$ of distinct vertices of $G$ is defined as 
$$\dil_G(p, q) := \frac{d_G(p, q)}{|pq|},$$
where $|pq|$ is the Euclidean distance between $p$ and $q$. The \emph{dilation} of $G$ is the largest dilation between two vertices of $G$,
$$\dil(G) := \sup_{p\neq q} \dil_G(p, q).$$
In other words, the dilation of $G$ is the least $t\geq 1$ such that, for any $p$ and $q$ in $P$, the graph distance $d_G(p, q)$ is at most $t$ times the Euclidean distance $|pq|$.

There has been extensive research on geometric graphs with low dilation which also satisfy other sparseness properties. We recall some of these results here, focussing on the following sparseness properties: being plane and having small maximum degree. We refer the reader to the survey by Bose and Smid~\cite{survey} for more details and related problems.

For some constant $C>0$, the following holds. For any finite point set $P$, there is a \emph{plane} geometric graph on $P$ with dilation at most $C$. The current best known constant $C$ is due to Xia~\cite{xia}, who gave a rather elaborate proof that $C = 1.998$ works. The previous record was $C=2$, a very elegant result of Chew~\cite{chew}. On the other hand, it is known that $C$ must be at least 1.4336~\cite{nous}. The best possible constant $C$ is conjectured to be very close to this lower bound.
	
The previous result still holds (for a different value of $C$) if we replace the condition of being plane by that of having \emph{maximum degree 3}~\cite{das}. By considering points arranged in a grid, it is readily seen that $3$ is the lowest possible maximum degree for which the result holds.

It is natural to ask whether we can simultaneously require planarity \emph{and} small maximum degree. Define the \emph{degree-$k$ dilation of $P$} by 
$$\dil_k(P) := \inf_{\substack{\Delta(G) \leq k\\G\text{ plane}}} \dil(G),$$
where the infimum is taken over all plane geometric graphs on $P$ of maximum degree~$k$. Bose \textit{et al.}~\cite{bose} were the first to show the existence of a $k$ (namely $k=23$) such that the degree-$k$ dilation of every finite point set $P$ is bounded by an absolute constant. A lot of research has been done to determine the best possible value of $k$. Bonichon \textit{et al.}~\cite{bonichon} (and later Kanj \textit{et al.}~\cite{kanj}) proved that it is possible to take $k = 4$. It is a major open problem to reduce the maximum degree down to $3$.

Computing the exact value of $\dil_k(P)$ for a concrete point set $P$ is not easy in general, because the set of geometric graphs to consider is often very large. Upper bounds for the degree-$3$ dilation of special classes of point sets have been obtained by Biniaz \textit{et al.}~\cite{biniaz}. Let us mention their upper bound of $3\sqrt{2}$ for the degree-$3$ dilation of non-uniform rectangular grids. 

The square lattice $\Z^2$ and the hexagonal lattice $\Lambda_{\Delta} = \Z \oplus e^{i\pi/3}\Z$ are among the few nontrivial examples of point sets for which the values $\dil_k(\cdot)$ are known. The following values were obtained by Dumitrescu and Ghosh~\cite{dumitrescu}. 

\begin{equation*}
	\begin{array}{c||c|c|c|c}
		\dil_k(\Lambda ) & k = 3 & k = 4 & k=5 & k\geq 6\\
		\hline 
	 	\hline
	 	\Lambda = \Z^2 & * & \sqrt{2} & \sqrt{2} & \sqrt{2} \\
	 	\hline 
	 	\Lambda = \Lambda_{\Delta}& 1+\sqrt{3}& 2 & 2 & \frac{2}{\sqrt{3}}
	\end{array}
\end{equation*}

For $\dilt(\Z^2)$, Dumitrescu and Ghosh showed that ${1+\sqrt{2} \leq \dilt(\Z^2) \leq (7+5\sqrt{2})/\sqrt{29}}$. They later gave the improved upper bound ${\dilt(\Z^2) \leq (3+2\sqrt{2})/\sqrt{5}}$, which they conjectured to be tight~\cite{dumitrescu}. We disprove this conjecture by giving examples of degree-$3$ plane geometric graphs of $\Z^2$ with dilation $1+\sqrt{2}$. 

The lower bound ${\dilt(\Z^2) \geq 1+\sqrt{2}}$ is trivial. Indeed, let $G$ be a geometric graph on $\Z^2$ of maximum degree~$3$. Let $p$ be an arbitrary point in $\Z^2$ and let $q_1, \ldots, q_4$ be the points in $\Z^2$ with $|pq_i|=1$ (see \cref{fig:lowerBound}). Since $p$ has degree a most $3$, there is some $1\leq i\leq 4$ with $d_G(p, q_i) > 1$, and thus $\dil(G) \geq \dil_G(p, q_i) = d_G(p, q_i) \geq 1+\sqrt{2}$. We will see below that there do exist graphs $G$ which match this lower bound. 

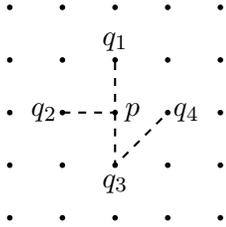
\begin{figure}[H]
\centering
\begin{tikzpicture}[scale = 0.7]
	\draw[dashed, thick] (0, 0) -- (0, -1) -- (1, 0); 
	\draw[dashed, thick] (-1, 0) -- (0, 0) -- (0, 1); 
	\foreach \x in {-2, -1, ..., 2} {
		\foreach \y in {-2, -1, ..., 2} {
			\fill (\x, \y) circle (1.5pt);
		}
	}
	\node[right] (0) at (0, 0) {$p$};
	\foreach \i in {1, 2, ..., 4} {
		\node[xshift = 7*cos(180*\i/2), yshift = 7*sin(180*\i/2)] 
			 (\i) at ({cos(180*\i/2)}, {sin(180*\i/2)}) {$q_{\i}$};
	}
\end{tikzpicture}
\caption{Lower bound on the degree-$3$ dilation of $\Z^2$.}
\label{fig:lowerBound}
\end{figure}

\begin{definition}
Let $\opti$ be the set of \emph{optimal} graphs, i.e.~the geometric graphs on $\Z^2$ of maximum degree $3$ which have dilation $1+\sqrt{2}$.
\end{definition}

\begin{definition}
We also define the set $\optiloc$ of \emph{locally optimal} graphs: the geometric graphs $G$ on $\Z^2$ of maximum degree $3$ which satisfy the dilation constraint $\dil_G(p, q)\leq 1+\sqrt{2}$ for every pair of vertices $(p, q)$ with $|pq| \leq \sqrt{5}$.
\end{definition}

We claim that the geometric graphs represented in \cref{fig:periodic} are optimal. Since they are periodic, it is easy to check that they are \emph{locally} optimal. That they indeed have dilation $1+\sqrt{2}$ directly follows from the following result.

\begin{restatable}[``Local-global principle'']{theorem}{main}\label{thm:main}
$\optiloc = \opti$.
\end{restatable}

The goal of this paper is to study the class of optimal graphs. Our main result is \cref{thm:main}, which characterizes the optimal graphs in terms of their structure in every ball of radius $\sqrt{5}$. It will be proved in \cref{sec:proofmain}, assuming a key lemma, \cref{lem:knight}. We will give a computer-assisted proof of \cref{lem:knight} in \cref{sec:CAP}.

\section{Degree-3 dilation of the square lattice}
\label{sec:proofmain}

We start this section by showing that the set of (locally) optimal graphs is very large. This seems to indicate that it might be difficult to give a fully explicit description of the set of optimal graphs, and thus makes \cref{thm:main} more interesting.

\begin{proposition}\label{prop:uncountable}
There are uncountably many locally optimal geometric graphs.
\end{proposition}

\begin{proof}[Proof]
For every countable sequence $(b_i)_{i\in \N}$ of zeroes and ones, we construct a geometric graph on $\Z^2$ as follows. Consider the periodic geometric graph in solid lines shown in \cref{fig:uncountable}, and let $(C_i)_{i\in \N}$ be an enumeration of the dashed circles. In the $i$-th circle, we add two vertical segments if $b_i = 1$ and two horizontal ones if $b_i = 0$. 

We obtain $2^{\aleph_0}$ degree-$3$ geometric graphs in this way. Verifying that these graphs are locally optimal is a finite check by ``almost periodicity''. This verification is performed by the python file \texttt{proposition\_2\_1.py}.
\end{proof}

\end{multicols}

\begin{figure}[ht]
\centering
\begin{tikzpicture}[scale = 0.5]
\tikzmath {\mywidth = 12; \myheight = 12;}
\begin{scope} 
	\clip (0, 0) rectangle (\mywidth, \myheight);
	\foreach \i in {-4, ..., 4} {
		\foreach \j in {-4, ..., 4} {
			\begin{scope}[shift = {(\i*4+\j*2, \j*(-3))}]
				\draw[thin] (0, 1) -- (1, 0) -- (4, 0) -- (4, 1) -- (5, 1) -- (3, 3);
				\draw[thin] (2, 0) -- (2, 1) -- (1, 2) -- (1, 1);
				\draw[thin] (1, 2) -- (2, 2) -- (1, 3);
				\draw[thin] (2, 1) -- (3, 1) -- (4, 0);
				\draw[thin] (3, 1) -- (3, 2);
				\draw[thin] (2, 2) -- (4, 2);
			\end{scope}
		}
	}
	\foreach \x in {-1, ..., 20} {
		\foreach \y in {-1, ..., 20} {
			\fill (\x, \y) circle (3pt);
		}
	}
	\begin{scope}[xshift = 5cm, yshift = 6cm]
		\draw[very thick, dotted, ->] (0, 0) -- (2, -3);
		\draw[very thick, dotted, ->] (0, 0) -- (2, 3);
	\end{scope}
\end{scope}
\begin{scope}[xshift = 13cm] 
	\clip (0, 0) rectangle (\mywidth, \myheight);
	\foreach \i in {-4, ..., 4} {
		\foreach \j in {-4, ..., 4} {
			\begin{scope}[shift = {(\i*3+\j*4, \i*2+\j*(-4))}]
				\draw[thin] (0, 0) -- (1, 0) -- (1, -1) -- (2, -1) -- (2, -2) -- (3, -2) -- (3, -3) -- (4, -3);
				\draw[thin] (4, -3) -- (4, -4) -- (5, -3) -- (6, -2) -- (7, -1);
				\draw[thin] (1, -1) -- (2, 0) -- (2, 1) -- (1, 1);
				\draw[thin] (3, -3) -- (4, -2) -- (4, -1) -- (4, 0) -- (3, -1) -- (2, -2);
				\draw[thin] (6, -2) -- (5, -2) -- (5, -1) -- (6, 0);
				\draw[thin] (4, -2) -- (5, -2);
				\draw[thin] (4, -1) -- (5, -1);
				\draw[thin] (4, 0) -- (5, 1);
				\draw[thin] (2, 0) -- (3, 0);
				\draw[thin] (2, 1) -- (3, 1);
				\draw[thin] (3, -1) -- (3, 0) -- (3, 1) -- (4, 2);
			\end{scope}
		}
	}
	\foreach \x in {-1, ..., 20} {
		\foreach \y in {-1, ..., 20} {
			\fill (\x, \y) circle (3pt);
		}
	}
	\begin{scope}[xshift = 4cm, yshift = 6cm]
		\draw[very thick, dotted, ->] (0, 0) -- (4, -4);
		\draw[very thick, dotted, ->] (0, 0) -- (3, 2);
	\end{scope}
\end{scope}
\begin{scope}[xshift = 26cm] 
	\clip (0, 0) rectangle (\mywidth, \myheight);
	\foreach \i in {-4, ..., 4} {
		\foreach \j in {-4, ..., 4} {
			\begin{scope}[shift = {(\i*2+\j*4, \i*3+\j*(-3))}]
				\draw[thin] (0, 0) -- (3, 3) -- (3, 2) -- (4,2) -- (4, 1) -- (5, 1) -- (5, 0) -- (6, 1) -- (6, 0) -- (7, 0);
				\draw[thin] (2, 2) -- (2, 0) -- (3, 0) -- (3, -1) -- (4, 0) -- (4, -1) -- (5, -1) -- (5, -2);
				\draw[thin] (4, -1) -- (4, -2);
				\draw[thin] (2, -2) -- (3, -1);
				\draw[thin] (1, -1) -- (2, 0);
				\draw[thin] (5, -1) -- (5, 0);
				\draw[thin] (3, 1) -- (4, 2);
				\draw[thin] (4, 0) -- (5, 1);
				\draw[thin] (2, 1) -- (3, 1) -- (3, 0);
			\end{scope}
		}
	}
	\foreach \x in {-1, ..., 20} {
		\foreach \y in {-1, ..., 20} {
			\fill (\x, \y) circle (3pt);
		}
	}
	\begin{scope}[xshift = 4cm, yshift = 6cm]
		\draw[very thick, dotted, ->] (0, 0) -- (2, 3);
		\draw[very thick, dotted, ->] (0, 0) -- (4, -3);
	\end{scope}
\end{scope}
\draw[very thick] (0, 0) rectangle (\mywidth, \myheight);
\draw[very thick] (13, 0) rectangle (\mywidth+13, \myheight);
\draw[very thick] (26, 0) rectangle (\mywidth+26, \myheight);
\end{tikzpicture}
\caption{Examples of periodic degree-$3$ spanners of $\Z^2$ with dilation $1+\sqrt{2}$.}
\label{fig:periodic}
\end{figure}
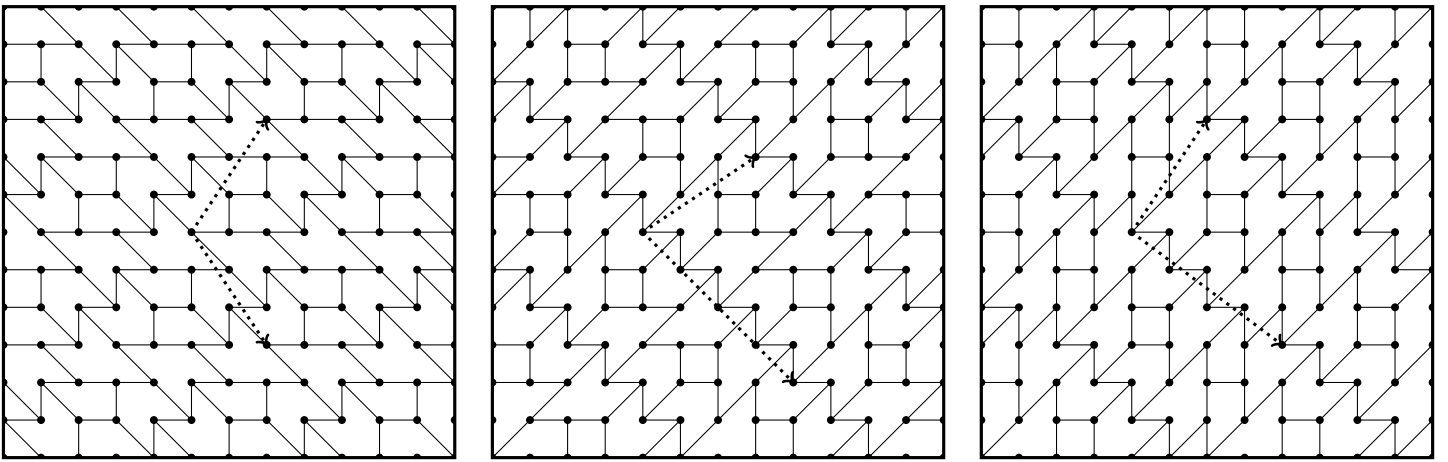

\begin{multicols}{2}

\begin{figure}[H]
\vspace*{0.1cm}
\begin{tikzpicture}[scale = 0.55]
\begin{scope}
	\clip (0, 1) rectangle (15, 16);
	\foreach \i/\j in {-12/0, -12/1, -12/6, -12/9, -12/10, -11/-1, -11/0, -11/2, -11/4, -11/5, -11/6, -11/7, -11/8, -11/9, -11/10, -10/3, -10/4, -10/5, -10/6, -10/7, -10/8, -10/11, -9/-1, -9/0, -9/2, -9/3, -9/4, -9/8, -9/11, -8/-1, -8/0, -8/4, -8/5, -8/6, -8/7, -8/9, -7/-1, -7/1, -7/2, -7/3, -7/5, -7/8, -7/9, -6/0, -6/1, -6/3, -6/5, -6/7, -6/9, -6/11, -5/-1, -5/1, -5/2, -5/4, -5/6, -5/7, -5/9, -5/10, -4/0, -4/2, -4/4, -4/6, -4/7, -4/8, -4/9, -4/11, -3/3, -3/4, -3/7, -3/8, -3/10, -2/-1, -2/0, -2/1, -2/2, -2/4, -2/7, -2/9, -2/10, -2/11, -1/-1, -1/0, -1/2, -1/6, -1/7, -1/10, 0/-1, 0/2, 0/3, 0/4, 0/7, 0/9, 0/10, 0/11}
	 {
			\begin{scope}[shift = {(\i*2+\j*5, \i*(-2))}]
				\draw[thin] (0, 0) -- (1, 1) -- (1, 2) -- (2, 2) -- (3, 3) -- (3, 2) -- (4, 2) -- (4, 1) -- (5, 1);
				\draw[thin] (1, 1) -- (2, 1) -- (2, 2);
				\draw[thin] (1, -1) -- (4, 2);
				\draw[very thick, dashed] (2, 1) -- (2, 0); 
				\draw[very thick, dashed] (3, 1) -- (3, 0); 
				\node (1) at (2.5, 0.5) {$1$};
				\draw[dashed] (1.center) circle (0.9cm);
			\end{scope}
	}
	\foreach \i/\j in {-12/-1, -12/2, -12/3, -12/4, -12/5, -12/7, -12/8, -12/11, -11/1, -11/3, -11/11, -10/-1, -10/0, -10/1, -10/2, -10/9, -10/10, -9/1, -9/5, -9/6, -9/7, -9/9, -9/10, -8/1, -8/2, -8/3, -8/8, -8/10, -8/11, -7/0, -7/4, -7/6, -7/7, -7/10, -7/11, -6/-1, -6/2, -6/4, -6/6, -6/8, -6/10, -5/0, -5/3, -5/5, -5/8, -5/11, -4/-1, -4/1, -4/3, -4/5, -4/10, -3/-1, -3/0, -3/1, -3/2, -3/5, -3/6, -3/9, -3/11, -2/3, -2/5, -2/6, -2/8, -1/1, -1/3, -1/4, -1/5, -1/8, -1/9, -1/11, 0/0, 0/1, 0/5, 0/6, 0/8}
	{
			\begin{scope}[shift = {(\i*2+\j*5, \i*(-2))}]
				\draw[thin] (0, 0) -- (1, 1) -- (1, 2) -- (2, 2) -- (3, 3) -- (3, 2) -- (4, 2) -- (4, 1) -- (5, 1);
				\draw[thin] (1, 1) -- (2, 1) -- (2, 2);
				\draw[thin] (1, -1) -- (4, 2);
				\draw[very thick, dashed] (2, 1) -- (3, 1); 
				\draw[very thick, dashed] (2, 0) -- (3, 0); 
				\node (1) at (2.5, 0.5) {$0$};
				\draw[dashed] (1.center) circle (0.9cm);
			\end{scope}
	}
	\foreach \x in {-1, ..., 20} {
		\foreach \y in {-1, ..., 20} {
			\fill (\x, \y) circle (3pt);
		}
	}
\end{scope}
\draw[very thick] (0, 1) rectangle (15, 16);
\end{tikzpicture}
\caption{Uncountably many locally optimal geometric graphs.}
\label{fig:uncountable}
\end{figure}
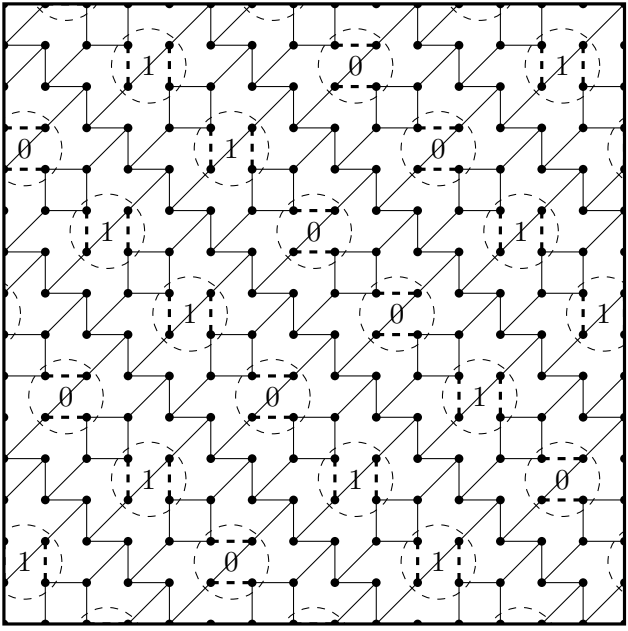

Assuming \cref{thm:main}, the geometric graphs we constructed in the previous proof have dilation $1+\sqrt{2}$.

We state the following key fact, for which we will give a computer-assisted proof in \cref{sec:CAP}.

\begin{lemma}[``Dilation boost'']\label{lem:knight}
Let $G \in \optiloc$. If $p, q\in \Z^2$ are such that $|pq| = \sqrt{5}$, then $d_G(p, q) \leq 3 + \sqrt{2}$.
\end{lemma}

\begin{remark}
The definition of $\optiloc$ only gives, a priori, that $d_G(p, q) \leq (1+\sqrt{2})\sqrt{5} \approx 5.40$. \Cref{lem:knight} improves this to $d_G(p, q) \leq 3+\sqrt{2} \approx 4.41$. 
\end{remark}

\begin{lemma}\label{lem:knightDistance}
Define an undirected weighted (non-geometric) graph $H$ with vertex set $\Z^2$ as follows: for $a, b\in \Z^2$,
\begin{itemize}
	\item if $|ab| \in \{1, \sqrt{2}\}$, then there is an edge between $a$ and $b$ of weight $(1+\sqrt{2})|ab|$;
	\item if $|ab| = \sqrt{5}$, then there is an edge between $a$ and $b$ of weight $3+\sqrt{2}$;
	\item otherwise, there is no edge between $a$ and $b$.
\end{itemize}
Then, for all $p, q\in \Z^2$, we have
\begin{equation}
\label{eq:dist}
d_{H}(p, q) \leq (1+\sqrt{2})|pq|.
\end{equation}
\end{lemma}

\begin{proof}
By translation invariance of $H$, we may assume that $p = (0, 0)$ in \cref{eq:dist}. The python file \texttt{lemma\_2\_4.py} checks that \cref{eq:dist} holds for $q\in \Z^2$ with $|pq| < 5\sqrt{2}$.

Suppose by contradiction that \cref{eq:dist} does not hold for some $q\in \Z^2$. Choose such a point $q$ with $|pq|$ minimal. By the symmetries of $H$, we may assume that $q$ lies in the first octant, so $q = (x, y)$ with $5\leq y \leq x$. 

We know that the point $r=(x-2, y-1)$ satisfies \cref{eq:dist}, since $|pr|< |pq|$. Therefore,
\begin{align*}
d_H(p, q) &\leq d_H(p, r) + d_H(r, q)  \\
&\leq (1+\sqrt{2})\sqrt{(x-2)^2+(y-1)^2} + (3+\sqrt{2})\\
&\leq (1+\sqrt{2})\sqrt{x^2+y^2},
\end{align*}
where the last inequality is true in the region $5\leq y \leq x$ by elementary calculus. This contradicts our assumption that $q$ does not satisfy \cref{eq:dist}.
\end{proof}

\main*

\begin{proof}[Proof of \cref{thm:main}, assuming \cref{lem:knight}]
Let $G$ be a locally optimal graph and let $H$ be the weighted graph defined in \cref{lem:knightDistance}. By definition of $\optiloc$ and \cref{lem:knight}, we see that ${d_G(p, q) \leq d_{H}(p, q)}$ for every $p, q\in \Z^2$. By \cref{lem:knightDistance}, we conclude that ${d_G(p, q)\leq (1+\sqrt{2})|pq|}$ for all $p, q\in \Z^2$, i.e.~$G\in \opti$.
\end{proof}

\section{Proof of the dilation boost}
\label{sec:CAP}

This section is devoted to the proof of \cref{lem:knight}. We start with an easy observation: there are only short edges in a locally optimal graph. 

\begin{lemma}\label{lem:short}
The edges of every $G \in \optiloc$ are of length $1$ or $\sqrt{2}$. 
\end{lemma}

\begin{proof}
Suppose that there is an edge in $G$ of length greater than $\sqrt{2}$. Without loss of generality, this edge may be assumed to have endpoints $a = (0, 0)$ and $b = (i, j)$ for some $1\leq i < j$. 

Assume for the moment that $j > 2$. Consider the points $p = (0, 1)$ and $q = (1, 1)$. Since $G \in \optiloc$, we need to have $d_G(p, q) \leq 1+\sqrt{2}$. As $G$ is plane, this forces the segments $pa$ and $aq$ to be edges of $G$. However, we also need to have $d_G(a, r) \leq 1+\sqrt{2}$, where $r = (0, -1)$. This is not possible since $a$ already has degree $3$. 

If $j=2$, i.e.~$b = (1, 2)$, the same reasoning applies, exchanging the roles of $a$ and $b$ if necessary.
\end{proof}

\begin{figure}[H]
\begin{tikzpicture}[scale = 0.7]
	\draw[thick] (0, 0) -- (2, 3);
	\draw[thick] (0, 1) -- (0, 0) -- (1, 1); 
	\foreach \x in {-1, ..., 3} {
		\foreach \y in {-2, -1, ..., 4} {
			\fill (\x, \y) circle (1.5pt);
		}
	}
	\node[below=3pt, right] (0) at (0, 0) {$a$};
	\node[right] (1) at (2, 3) {$b$};
	\node[left] (2) at (0, 1) {$p$};
	\node[right] (3) at (1, 1) {$q$};
	\node[right] (4) at (0, -1) {$r$};
	
	\draw[dashed, thick, <->] (0, -1) arc (-150: -210: 1cm);
	\node[left = 2pt] (-1) at (0, -0.5) {$?$};
\end{tikzpicture}
\caption{Proof of \cref{lem:short}.}
\label{fig:shortEdges}
\end{figure}

The previous lemma says that some ``edge patterns'', namely edges of length greater than $\sqrt{2}$, cannot appear in a locally optimal graph. In \cref{lem:forbidden}, we give two more such patterns. This time, the proof is computer-assisted.

\begin{lemma}\label{lem:forbidden}
Let $G \in \optiloc$ and let $H_1, H_2$ be the edge configurations in \cref{fig:forbidden}. Then, neither $H_1$ nor $H_2$ (nor any translation, rotation or reflection of one of these two configurations) is a subgraph of $G$.
\end{lemma}

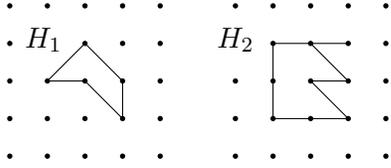
\begin{figure}[H]
\centering
\begin{tikzpicture}
\begin{scope}[scale = 0.5, yshift = 1cm] 
	\clip (-2.1, -2.1) rectangle (2.1, 2.1);
	\foreach \x in {-2, ..., 2} {
			\foreach \y in {-2, ..., 2} {
				\fill (\x, \y) circle (2pt);
			}
		}
	\draw (0, 0) -- (-1, 0) -- (0, 1) -- (1, 0) -- (1, -1) -- (0, 0);
	\fill[color = white] (-1.1, 1.1) circle (0.6cm);
	\node (0) at (-1.1, 1.1) {$H_1$};
\end{scope}
\begin{scope}[scale = 0.5, xshift = 6cm] 
	\foreach \x in {-2, ..., 2} {
			\foreach \y in {-1, ..., 3} {
				\fill (\x, \y) circle (2pt);
			}
		}
	\draw (1, 2) -- (-1, 2) -- (-1, 0) -- (1, 0) -- (0, 1) -- (1, 1) -- (0,2);
	\fill[color = white] (-2, 2.1) circle (0.6cm);
	\node (0) at (-2, 2.1) {$H_2$};
\end{scope}
\end{tikzpicture}
\caption{Two configurations that cannot appear in a locally optimal graph.}
\label{fig:forbidden}
\end{figure}

Suppose that we wish to prove that a given set $S$ of edges on $\Z^2$ is never contained in a locally optimal graph (for example, $S = H_1$ or $H_2$). We start with some definitions.

\begin{definition} 
A \emph{close pair} is a pair $(p, q)$ of points in~$\Z^2$ such that $|pq| \leq \sqrt{5}$.
\end{definition}

\begin{definition}
Let $(p, q)$ be a close pair, and let $\gamma = (v_1v_2, v_2v_3, \ldots, v_{n-1}v_n)$ be a path between $v_1 = p$ and $v_n = q$ (each intermediate vertex $v_i$ is in $\Z^2$, the edges $v_iv_{i+1}$ are not necessarily in $S$). We say that $\gamma$ is an \emph{$S$-admissible path between $p$ and $q$} if the following conditions are verified: 
\begin{itemize}
	\item each edge $v_iv_{i+1}$ has length $1$ or $\sqrt{2}$;
	\item $(\Z^2, S \cup \gamma)$ is a plane geometric graph of maximum degree at most $3$; 
	\item the length of $\gamma$ is at most $(1+\sqrt{2})|pq|$.
\end{itemize}
\end{definition}

For any close pair $(p, q)$, exactly one of the following cases must occur. 

\begin{enumerate}
	\item \textsc{Contradiction.} There is no $S$-admissible path between $p$ and $q$.	
	\item \textsc{Satisfaction.} There is at least one $S$-admissible path between $p$ and $q$ which is entirely contained in $S$.
	\item \textsc{Deduction.} There is exactly one $S$-admissible path between $p$ and $q$, and this path is not entirely contained in $S$.	
	\item \textsc{Exploration.} There are several $S$-admissible paths between $p$ and $q$, none of which is entirely contained in $S$.
\end{enumerate}

With this terminology, we can now present the backtracking algorithm that we will use. If $G$ is a locally optimal graph containing $S$, every close pair must have an $S$-admissible path $\gamma$ consisting of edges of $G$. Using this fact we can, starting from $S$, try all possible ways of reconstructing $G$ and hope to eventually find a contradiction in all cases. More precisely, consider \cref{alg:version1}.

\begin{algorithm}[H]
	\caption{Goal: prove that a set $S_0$ of edges cannot be contained in a locally optimal graph}
	\label{alg:version1}
	\begin{algorithmic}[1]
		\Input $S_0$, a finite set of edges 
		\vspace*{0.1cm}
		\Function{Expand}{$S$} 
			\If{we detect at least one close pair \newline \hspace*{2.3em} with \textsc{Contradiction}} 
				\State\Return
			\ElsIf{we detect at least one close pair \newline \hspace*{4.5em} with \textsc{Deduction}}
				\Let{$(p, q)$}{any close pair with \textsc{Deduction}}
				\Let{$\gamma$}{the $S$-admissible path between $p$ and $q$}
				\State{\textsc{Expand}$(S\cup \gamma)$} 
			\Else 
				\Let{$(p, q)$}{any close pair} 
				\Let{$L$}{$\{S$-admissible paths between $p$ and $q\}$}
				\For{$\gamma$ in $L$} 
					\State\textsc{Expand}$(S\cup \gamma)$
				\EndFor
			\EndIf
		\EndFunction
		\State
		\State \textsc{Expand($S_0$)} \Comment{Terminates $\Rightarrow$ valid proof}
	\end{algorithmic}
\end{algorithm}

\begin{remark}
On lines [2:]\ and [4:]\ of \cref{alg:version1}, by ``we detect'', it is meant that the algorithm spots a close pair with the desired property, but it might not find any even if one exists. 
\end{remark}

\begin{claim}\label{claim}
If \cref{alg:version1} terminates with input $S$, then $S$ cannot be contained in a locally optimal graph.
\end{claim}

\begin{proof}[Proof of \cref{claim}]
Suppose by contradiction that \cref{alg:version1} terminates for a finite set of edges $S_0$ which is contained in a locally optimal graph $G$.

In the execution of \cref{alg:version1}, there is a finite number of calls to \textsc{Expand}. Consider \textsc{Expand}$(S)$, the last of these calls where the argument $S$ is a (finite) set of edges of $G$.

Since $G \in \optiloc$, any close pair $(p, q)$ has at least one $S$-admissible path $\gamma_{p, q}$ consisting of edges of $G$. In particular, there is no close pair with \textsc{Contradiction}, and either lines [5-7:]\ or [9-13:]\ will be executed. 

Let $(p, q)$ be the pair chosen on line [5:]\ or [9:]. By construction, there will be a call to \textsc{Expand}$(S\cup \gamma_{p, q})$ on line [7:]\ or [12:]. However, $S\cup \gamma_{p, q}$ is a finite set of edges of $G$, contradicting the assumption on $S$.
\end{proof}

\begin{remark} 
\Cref{claim} holds regardless of how the points $p$ and $q$ are chosen on lines [5:]\ and [9:]. In practice, on line [9:], it is important to choose the points $p$ and $q$ in a such a way that the algorithm terminates in a reasonable amount of time (or terminates at all\footnote{It may be the case that, given two different ways of choosing the pairs $(p, q)$ in \cref{alg:version1}, the algorithm terminates for one but not for the other (with the same input $S_0$ in both cases). However, if it terminates for some choices of pairs $(p, q)$, we are certain that $S_0$ cannot be contained in a locally optimal graph.}).
\end{remark}

\begin{proof}[Proof of \cref{lem:forbidden}]
We first apply (the implemented version of) \cref{alg:version1} with input $S_0 = H_1$ and see that the algorithm terminates. By \cref{claim}, \cref{lem:forbidden} is proved for $H_1$.

Having just proved that $H_1$ cannot be contained in a locally optimal graph, we may use a slightly modified method for $H_2$. We apply \cref{alg:version1} with input $S_0 = H_2$ after inserting the following test between lines [1:]\ and [2:]\ of \cref{alg:version1}.

\algrenewcommand\alglinenumber[1]{\footnotesize +#1:}
\vspace*{0.2cm}
\hspace*{0.65cm}
\begin{minipage}{0.4\textwidth}
        \begin{algorithmic}[1]
            	\If{we detect a copy of $H_1$ in $S$}
                  \State\Return
                \EndIf
        \end{algorithmic}
\end{minipage}
\vspace*{0.2cm}

Again, this modified version of \cref{alg:version1} terminates. Thus \cref{lem:forbidden} is proved for $H_2$.  
\end{proof}

A more general version of \cref{alg:version1} is implemented in the Python file \texttt{proof.py}. Further explanations will be given after \cref{alg:version2}.

\medbreak 

For the remainder of this section, fix, by contradiction, a locally optimal graph $G$ that violates the conclusion of \cref{lem:knight}. There exist $u, v\in \Z^2$ with $|uv| = \sqrt{5}$ that satisfy $d_G(u, v) > 3+\sqrt{2}$. Without loss of generality, we may assume that $u = (0, 0)$ and $v = (1, 2)$. 

\begin{lemma}\label{lem:paths}
Any shortest path in $G$ between $u$ and $v$ must be, up to symmetry, one of the four possibilities represented in \cref{fig:paths}.
\end{lemma}

\begin{proof}
Let $P$ be a shortest path between $u$ and $v$. By assumption, $3+\sqrt{2} < \mathrm{length}(P) \leq \sqrt{5}(1+\sqrt{2})$. This leaves only a small number of possibilities for $P$, which are enumerated by the Python program \texttt{lemma\_3\_8.py}. 

Not all paths whose lengths are in this range can actually be \emph{shortest} paths between $u$ and $v$. Some paths can be discarded using the constraint $\dil_G(x, y) \leq 1+\sqrt{2}$ not only for $u$ and $v$ but also for intermediate vertices in the path $P$ (execute \texttt{lemma\_3\_8.py} for the details).
\end{proof}

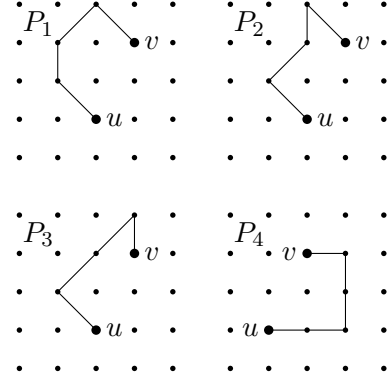
\begin{figure}[H]
\centering
\vspace*{0.1cm}
\begin{tikzpicture}[scale = 1.02]
\begin{scope}[scale = 0.5, xshift = 0cm] 
	\foreach \x in {-2, ..., 2} {
			\foreach \y in {-1, ..., 3} {
				\fill (\x, \y) circle (2pt);
			}
		}
	\draw (0, 0) -- (-1, 1) -- (-1, 2) -- (0, 3) -- (1, 2);
	\fill (0, 0) circle (3.5pt);
	\node[right] (0) at (0, 0) {$u$};
	\fill (1, 2) circle (3.5pt);
	\node[right] (0) at (1, 2) {$v$};
	\node (0) at (-1.5, 2.5) {$P_1$};
\end{scope}
\begin{scope}[scale = 0.5, xshift = 5.5cm] 
	\foreach \x in {-2, ..., 2} {
			\foreach \y in {-1, ..., 3} {
				\fill (\x, \y) circle (2pt);
			}
		}
	\draw (0, 0) -- (-1, 1) -- (0, 2) -- (0, 3) -- (1, 2);
	\fill (0, 0) circle (3.5pt);
	\node[right] (0) at (0, 0) {$u$};
	\fill (1, 2) circle (3.5pt);
	\node[right] (0) at (1, 2) {$v$};	
	\node (0) at (-1.5, 2.5) {$P_2$};
\end{scope}
\begin{scope}[scale = 0.5, yshift = -5.5cm] 
	\foreach \x in {-2, ..., 2} {
			\foreach \y in {-1, ..., 3} {
				\fill (\x, \y) circle (2pt);
			}
		}
	\draw (0, 0) -- (-1, 1) -- (1, 3) -- (1, 2);
	\fill (0, 0) circle (3.5pt);
	\node[right] (0) at (0, 0) {$u$};
	\fill (1, 2) circle (3.5pt);
	\node[right] (0) at (1, 2) {$v$};
	\node (0) at (-1.5, 2.5) {$P_3$};
\end{scope}
\begin{scope}[scale = 0.5, xshift = 5.5cm, yshift = -5.5cm] 
	\foreach \x in {-2, ..., 2} {
			\foreach \y in {-1, ..., 3} {
				\fill (\x, \y) circle (2pt);
			}
		}
	\draw (-1, 0) -- (1, 0) -- (1, 2) -- (0, 2);
	\fill (-1, 0) circle (3.5pt);
	\node[left] (0) at (-1, 0) {$u$};
	\fill (0, 2) circle (3.5pt);
	\node[left] (0) at (0, 2) {$v$};
	\node (0) at (-1.5, 2.5) {$P_4$};
\end{scope}
\end{tikzpicture}
\caption{List of the possible shortest paths between $u$ and $v$, up to symmetry.}
\label{fig:paths}
\end{figure}

We can now adapt \cref{alg:version1} to give a computer-assisted proof of the dilation boost.

\begin{algorithm}[H]
	\caption{Goal: prove that there is no ${G\in \optiloc}$ that contains a set $S_0$ of edges and such that a certain constraint of the form $d_G(u, v) \geq c$ is satisfied}
	\label{alg:version2}
	\begin{algorithmic}[1]
		\Input {$S_0$, a finite set of edges}
		\PhantomInput{two points $u, v\in \Z^2$ and a constant $c>0$}
		\vspace*{0.1cm}
		\Function{Expand}{$S$} 
			\If{we detect that $d_G(u, v) < c$ \newline \hspace*{2.3em} for all $G\in \optiloc$ containing $S$}
				\State\Return
			\ElsIf{we detect a copy of $H_1$ or $H_2$ in $S$}
				\State\Return
			\ElsIf{we detect at least one close pair \newline \hspace*{4.5em} with \textsc{Contradiction}}         	  
				\State\Return
			\ElsIf{we detect at least one close pair \newline \hspace*{4.5em} with \textsc{Deduction}} 
				\Let{$(p, q)$}{any close pair with \textsc{Deduction}}
				\Let{$\gamma$}{the $S$-admissible path between $p$ and $q$}
				\State{\textsc{Expand}$(S\cup \gamma)$} 
			\Else        	 
				\Let{$(p, q)$}{any close pair}
				\Let{$L$}{$\{S$-admissible paths between $p$ and $q\}$}
				\For{$\gamma$ in $L$} 
					\State\textsc{Expand}$(S\cup \gamma)$
				\EndFor
			\EndIf
		\EndFunction
		\State
		\State \textsc{Expand($S_0$)} \Comment{Terminates $\Rightarrow$ valid proof}
	\end{algorithmic}
\end{algorithm}

\begin{proof}[Proof of \cref{lem:knight}]
Let $P$ be a shortest path between $u$ and $v$. By \cref{lem:paths}, we may suppose that $P$ is one of $P_1, \ldots, P_4$. Suppose $P = P_i$ for some $1\leq i\leq 4$. Let $c_i = \mathrm{length}(P_i)$.

We want to prove that there is no locally optimal graph $G$ containing $P_i$ such that $d_G(u, v) \geq c_i$. We can use the same method as in the proof of \cref{lem:forbidden}, adding the constraint $d_G(u, v) \geq c_i$ to the algorithm. 

Concretely, we execute \cref{alg:version2} with $S_0 = P_i$, $u = (0, 0)$, $v = (1, 2)$ and $c = c_i$, and we do so for $1\leq i\leq 4$. In each case, we observe that the algorithm terminates.

\Cref{alg:version1,alg:version2} have a common implementation in the Python file \texttt{proof.py}. The file \texttt{interface.py} allows the reader to visualize the proof in real time, while \texttt{launch.py} contains the input data. 

Please check \texttt{README.md} to see how to execute the different parts of the proof with customized visualization options. Implementation details and configuration instructions may also be found in the \texttt{README.md} file. 

All files can be found on GitHub at \url{https://git.io/JTiCD} or on arXiv at \url{https://www.arxiv.org/src/2010.13473/anc}.
\end{proof}

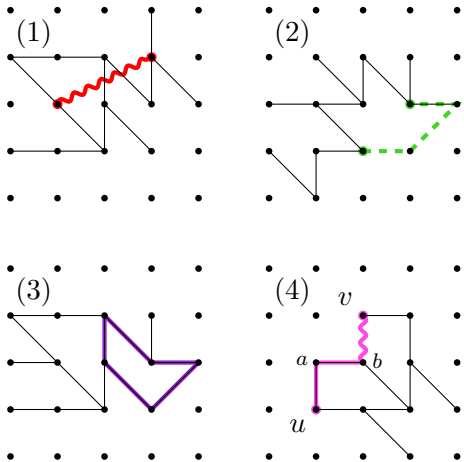
\begin{figure}[H]
\centering
\begin{tikzpicture}[scale=1.25]
\begin{scope}[scale = 0.5, xshift = 0cm] 
	\tikzset{decoration={snake,amplitude=.4mm,segment length=2mm,
                       post length=0mm,pre length=0mm}}
         \draw[decorate, color = red, ultra thick] (1, 1) -- (3, 2);
	\fill[color = red] (3, 2) circle (3pt);
	\fill[color = red] (1, 1) circle (3pt);
	\foreach \x in {0, ..., 4} {
			\foreach \y in {-1, ..., 3} {
				\fill (\x, \y) circle (2pt);
			}
		}
	\draw (0, 0) -- (1, 0) -- (2, 0) -- (1, 1) -- (0, 2) -- (1, 2) -- (2, 2) -- (2, 1) -- (2, 0);
	\draw (2, 1) -- (3, 0);
	\draw (2, 2) -- (3, 1) -- (3, 2) -- (3, 3);
	\draw (3, 2) -- (4, 1);	
	\node (0) at (0.5, 2.5) {$(1)$};
\end{scope}
\begin{scope}[scale = 0.5, xshift = 6.5cm] 
	\draw[color = {rgb,255:red,60; green,216; blue,29}, ultra thick, dashed] (1, 0) -- (2, 0) -- (3, 1) -- (2, 1); 
	\fill[color = {rgb,255:red,60; green,216; blue,29}] (1, 0) circle (3pt);
	\fill[color = {rgb,255:red,60; green,216; blue,29}] (2, 1) circle (3pt);	
	\foreach \x in {-1, ..., 3} {
			\foreach \y in {-1, ..., 3} {
				\fill (\x, \y) circle (2pt);
			}
		}
	\draw (0, 0) -- (1, 0) -- (0, 1) -- (1, 1) -- (0, 2);
	\draw (1, 1) -- (1, 2) -- (2, 1) -- (2, 2);
	\draw (2, 1) -- (3, 1);
	\draw (-1, 1) -- (0, 1);
	\draw (0, 0) -- (0, -1) -- (-1, 0);
	\node (0) at (-0.5, 2.5) {$(2)$};
\end{scope}
\begin{scope}[scale = 0.5, xshift = 0cm, yshift = -5.5cm] 
	\draw[color = {rgb,255:red,150; green,35; blue,200}, ultra thick] (2, 2) -- (2, 1) -- (3, 0) -- (4, 1); 	
	\draw[color = {rgb,255:red,150; green,35; blue,200}, ultra thick] (4, 1) -- (3, 1) -- (2, 2);
	\draw (0, 0) -- (1, 0) -- (2, 0) -- (2, 1) -- (2, 2) -- (1, 2) -- (0, 2) -- (1, 1) -- (2, 0);
	\draw (0, 1) -- (1, 1);
	\draw (2, 1) -- (3, 0) -- (4, 1) -- (3, 1) -- (3, 2);
	\draw (3, 1) -- (2, 2);
	\foreach \x in {0, ..., 4} {
			\foreach \y in {-1, ..., 3} {
				\fill (\x, \y) circle (2pt);
			}
		}
	\node (0) at (0.5, 2.5) {$(3)$};
\end{scope}
\begin{scope}[scale = 0.5, xshift = 6.5cm, yshift = -5.5cm] 
	\draw[color = {rgb,255:red,255; green,75; blue,230}, ultra thick] (0, 0) -- (0, 1) -- (1, 1);
	\tikzset{decoration={snake,amplitude=.4mm,segment length=2mm,
                       post length=0mm,pre length=0mm}}
         \draw[decorate, color = {rgb,255:red,255; green,75; blue,230}, ultra thick] (1, 1) -- (1, 2);
	\fill[color = {rgb,255:red,255; green,75; blue,230}] (0, 0) circle (3pt);
	\node[below left] (0) at (0, 0) {$u$};
	\fill[color = {rgb,255:red,255; green,75; blue,230}] (1, 2) circle (3pt);
	\node[above left] (1) at (1, 2) {$v$};
	\node (2) at (-0.3, 1.05) {\scriptsize{$a$}};
	\node (3) at (1.3, 1.05) {\scriptsize{$b$}};
	\foreach \x in {-1, ..., 3} {
			\foreach \y in {-1, ..., 3} {
				\fill (\x, \y) circle (2pt);
			}
		}
	\draw (0, 0) -- (1, 0) -- (2, 0) -- (2, 1) -- (2, 2) -- (1, 2);
	\draw (0, 0) -- (0, 1) -- (1, 1) -- (2, 0);
	\draw (1, 0) -- (2, -1);
	\draw (2, 1) -- (3, 0);
	\node (0) at (-0.5, 2.5) {$(4)$};
\end{scope}
\end{tikzpicture}
\caption{Graphical Interface examples.}
\label{fig:gui}
\end{figure}

We end this section by explaining the visualization produced by the \texttt{interface.py} file. \Cref{fig:gui} shows an example for each type of behavior. The segments in black form the current edge set $S$.

\begin{enumerate}
	\item \textsc{Contradiction}: there is no $S$-admissible path between the two endpoints of the red string. This corresponds to line~[6:]\ of \cref{alg:version2}.
	\item \textsc{Deduction}: the path represented in green is the only $S$-admissible path between the two endpoints (line~[8:]\ of \cref{alg:version2}).
	\item We detect a (rotated) copy of $H_1$ in $S$ (line [4:]).
	\item The condition on line [2:]\ of \cref{alg:version2} is verified with $c = 5$. Indeed, let $G$ be a locally optimal graph containing $S$. There is already a path (through $a$) between $u$ and $b$ in $S$ of length $2$, and $d_G(b, v) \leq 1+\sqrt{2}$ as $G\in \optiloc$. Thus, we must have $d_G(u, v)\leq 3+\sqrt{2} < 5$.
\end{enumerate}

\bibliography{article}

\providecommand{\bysame}{\leavevmode\hbox to3em{\hrulefill}\thinspace}
\providecommand{\MR}{\relax\ifhmode\unskip\space\fi MR }
\providecommand{\MRhref}[2]{%
  \href{http://www.ams.org/mathscinet-getitem?mr=#1}{#2}
}
\providecommand{\href}[2]{#2}
\begin{thebibliography}{10}

\bibitem{biniaz}
A.~Biniaz, Prosenjit Bose, J.-L. De~Carufel, C.~Gavoille, Anil Maheshwari, and
  Michiel Smid, \emph{Towards plane spanners of degree 3}, Leibniz
  International Proceedings in Informatics, December 2016, pp.~19.1--19.14.

\bibitem{bonichon}
Nicolas Bonichon, Iyad Kanj, Ljubomir Perkovi{\'c}, and Ge~Xia, \emph{There are
  plane spanners of degree 4 and moderate stretch factor}, Discrete \&
  Computational Geometry \textbf{53} (2015), no.~3, 514--546.

\bibitem{bose}
Prosenjit Bose, Joachim Gudmundsson, and Michiel Smid, \emph{Constructing plane
  spanners of bounded degree and low weight}, Algorithmica \textbf{42} (2005),
  no.~3-4, 249--264.

\bibitem{survey}
Prosenjit Bose and Michiel Smid, \emph{On plane geometric spanners: A survey
  and open problems}, Computational Geometry \textbf{46} (2013), no.~7,
  818--830.

\bibitem{chew}
L~Paul Chew, \emph{There are planar graphs almost as good as the complete
  graph}, Journal of Computer and System Sciences \textbf{39} (1989), no.~2,
  205--219.

\bibitem{das}
Gautam Das and Paul~J Heffernan, \emph{Constructing degree-3 spanners with
  other sparseness properties}, International Journal of Foundations of
  Computer Science \textbf{7} (1996), no.~02, 121--135.

\bibitem{dumitrescu}
Adrian Dumitrescu and Anirban Ghosh, \emph{Lattice spanners of low degree},
  Discrete Mathematics, Algorithms and Applications \textbf{8} (2016), no.~03,
  1650051.

\bibitem{nous}
Damien Galant and C\'{e}dric Pilatte, \emph{Dilation of limit triangulations},
  (In preparation).

\bibitem{kanj}
Iyad Kanj, Ljubomir Perkovi{\'c}, and Duru T{\"u}rko{\u{g}}lu, \emph{Degree
  four plane spanners: Simpler and better}, Journal of Computational Geometry
  \textbf{8} (2017), no.~2.

\bibitem{xia}
Ge~Xia, \emph{The stretch factor of the delaunay triangulation is less than
  1.998}, SIAM Journal on Computing \textbf{42} (2013), no.~4, 1620--1659.

\end{thebibliography}
\bibliographystyle{amsplain}

\end{multicols}

\end{document}